\documentclass[12pt]{article}    

\usepackage[margin=0.75in]{geometry} 
\usepackage{amsmath,amssymb,amsthm}

\newtheorem{theorem}{Theorem}[section]
\newtheorem{proposition}[theorem]{Proposition}
\newtheorem{lemma}[theorem]{Lemma}
\newtheorem{corollary}[theorem]{Corollary}

\theoremstyle{remark}

\newcommand{\be}{\begin{equation}}
\newcommand{\ee}{\end{equation}}

\newcommand{\bea}{\begin{eqnarray}}
\newcommand{\eea}{\end{eqnarray}}

\numberwithin{equation}{section}
\begin{document}

\title{Painlev\'{e} Transcendents and the Hankel Determinants Generated by a Discontinuous Gaussian Weight}
\author{Chao Min\thanks{School of Mathematical Sciences, Huaqiao University, Quanzhou 362021, China; email: chaomin@hqu.edu.cn}\: and Yang Chen\thanks{Correspondence to: Yang Chen, Department of Mathematics, Faculty of Science and Technology, University of Macau, Macau, China;
email: yangbrookchen@yahoo.co.uk}}


\date{\today}
\maketitle
\begin{abstract}
This paper studies the Hankel determinants generated by a discontinuous Gaussian weight with one and two jumps. It is an extension of Chen and Pruessner \cite{Chen2005}, in which they studied the discontinuous Gaussian weight with a single jump. By using the ladder operator approach, we obtain a series of difference and differential equations to describe the Hankel determinant for the single jump case. These equations include the Chazy II equation, continuous and discrete Painlev\'{e} IV. In addition, we consider the large $n$ behavior of the corresponding orthogonal polynomials and prove that they satisfy the biconfluent Heun equation. We also consider the jump at the edge under a double scaling, from which a Painlev\'{e} XXXIV appeared. Furthermore, we study the Gaussian weight with two jumps, and show that a quantity related to the Hankel determinant satisfies a two variables' generalization of the Jimbo-Miwa-Okamoto $\sigma$ form of the Painlev\'{e} IV.
\end{abstract}

$\mathbf{Keywords}$: Hankel determinants; Random matrices; Orthogonal polynomials;

Ladder operators; Painlev\'{e} transcendents.

$\mathbf{Mathematics\:\: Subject\:\: Classification\:\: 2010}$: 15B52, 42C05, 33E17

\section{Introduction}
In the theory of random matrix ensembles, the joint probability density function for the eigenvalues $\{x_j\}_{j=1}^{n}$ of $n\times n$ Hermitian matrices from an unitary ensemble is given by \cite{Mehta}
\be\label{jpdf}
P(x_{1},x_{2},\ldots,x_{n})\prod_{j=1}^{n}dx_j=\frac{1}{D_{n}[w_{0}]}\frac{1}{n!}\prod_{1\leq j<k\leq n}\left(x_{k}-x_{j}\right)^{2}
\prod_{j=1}^{n}w_{0}(x_{j})dx_j,
\ee
where $w_{0}(x)$ is a weight function supported on $[a,b]$ and the moments of $w_{0}(x)$, namely,
$$
\mu_j:=\int_{a}^{b}x^j w_{0}(x)dx,\;\;j=0,1,2,\ldots
$$
exist.
Here $D_{n}[w_{0}]$ is the normalization constant
$$
D_{n}[w_{0}]=\frac{1}{n!}\int_{[a,b]^{n}}\prod_{1\leq j<k\leq n}\left(x_{k}-x_{j}\right)^{2}
\prod_{j=1}^{n}w_{0}(x_{j})dx_{j}.
$$


In this paper, we take
$$
w_{0}(x)=\mathrm{e}^{-x^{2}},\;\; x\in \mathbb{R},
$$
which corresponds to the Gaussian unitary ensemble. In this case, $D_{n}[w_{0}]$ has the closed-form expression \cite{Mehta},
$$
D_{n}[w_{0}]=(2\pi)^{\frac{n}{2}}2^{-\frac{n^2}{2}}G(n+1),
$$
where $G(z)$ is the Barnes G-function, defined by
$$
G(z+1)=\Gamma(z)G(z),\;\;G(1)=1
$$
and $\Gamma(z)$ is the gamma function, $\Gamma(z):=\int_{0}^{\infty}t^{z-1}\mathrm{e}^{-t}dt,\;\mathrm{Re}(z)>0$.
Noting that, \cite{Mehta} also gives the closed-form expression $D_{n}[w_{0}]$ for other weight functions, such as the Laguerre weight, $w_{0}(x)=x^{\alpha}\:{\rm e}^{-x},\;\alpha>-1,\;\;x\in \mathbb{R}^+$ and the Jacobi weight,
$w_{0}(x)=(1-x)^{\alpha}(1+x)^{\beta},\;\alpha>-1,\;\beta>-1,\;x\in [-1,1]$.

Linear statistics is a ubiquitous statistical characteristic in random matrix theory \cite{Basor1997,Chen1998,Chen1994,Johansson,Min201601,Min201602}. A linear statistic is a linear sum of a certain function $g$ of the random variable $x_{j}$: $\sum_{j=1}^{n}g(x_j)$. In our paper, $\{x_{j},\;j=1,2,\ldots,n\}$ are the eigenvalues of Hermitian matrices. Whenever one encounters a linear statistic, it is useful to consider its ``exponential'' generating function, which is defined by the expectation of ${\rm e}^{\lambda\sum_{j=1}^{n}g(x_j)}$ over the joint probability distribution (\ref{jpdf}), where $\lambda$ is a parameter, that is,
$$
\mathbb{E}\left({\rm e}^{\lambda\:\sum_{j=1}^{n}g(x_j)}\right)
:=\frac{1}{D_{n}[w_{0}]}\frac{1}{n!}\int_{(-\infty,\infty)^{n}}\prod_{1\leq j<k\leq n}\left(x_{k}-x_{j}\right)^{2}\prod_{j=1}^{n}w_{0}(x_{j})\:{\rm e}^{\lambda\:g(x_j)}dx_{j}.
$$

More generally, we can consider
\be\label{exp}
\mathbb{E}\left(\prod_{j=1}^{n}f(x_j)\right)
:=\frac{1}{D_{n}[w_{0}]}\frac{1}{n!}\int_{(-\infty,\infty)^{n}}\prod_{1\leq j<k\leq n}\left(x_{k}-x_{j}\right)^{2}\prod_{j=1}^{n}w_{0}(x_{j})f(x_{j})dx_{j}.
\ee
In this paper, we investigate the case
\be\label{fx}
f(x)=A+B_{1}\theta(x-t_{1})+B_{2}\theta(x-t_{2}),\;\;t_{1}<t_{2},
\ee
where $\theta(x)$ is the Heaviside step function, i.e., $\theta(x)$ is 1 for $x>0$ and 0 otherwise, and $A, B_{1}, B_{2}$ are constants, $A\geq0,\;A+B_{1}\geq0,\;A+B_{1}+B_{2}\geq0$.
\\
Let
\be\label{w12}
w(x,t_{1},t_{2}):=w_{0}(x)(A+B_{1}\theta(x-t_{1})+B_{2}\theta(x-t_{2}))
\ee
and
$$
D_{n}[w]:=\frac{1}{n!}\int_{(-\infty,\infty)^{n}}\prod_{1\leq j<k\leq n}\left(x_{k}-x_{j}\right)^{2}
\prod_{j=1}^{n}w(x_{j},t_{1},t_{2})dx_{j}.
$$
We will see that $D_{n}[w]$ is the Hankel determinant generated by the weight $w(x,t_{1},t_{2})$.

The motivation of this paper comes from Chen and Pruessner \cite{Chen2005}, in which they studied the Hankel determinant generated by the Gaussian weight with a single jump. In addition, Basor and Chen \cite{Basor2009}, Chen and Zhang \cite{Chen2010} investigated the Hankel determinants for the Laguerre weight and Jacobi weight with a single jump, respectively. We would like to consider the discontinuous Gaussian weight (\ref{w12}). The single jump case is recovered if $B_{2}=0$, and we also study this case since we have some important results beyond \cite{Chen2005}. So this paper is divided into two parts, the first part is the single jump case and the other is the general case with two jumps.

We would like to point out that, there are three important special cases from (\ref{exp}) and (\ref{fx}): the first one is $A=0,\;B_{1}=1,\;B_{2}=0$, and this leads us to compute the probability that the smallest eigenvalue is greater than $t_{1}$; the second one is $A=1,\;B_{1}=-1,\;B_{2}=0$, and this allows us to compute the probability that the largest eigenvalue is less than $t_{1}$; the last one is $A=0,\;B_{1}=1,\;B_{2}=-1$ and is related to the probability of the eigenvalues lying in the interval $(t_{1},t_{2})$, which has been studied by Basor, Chen and Zhang \cite{Basor2012}. Moreover, Tracy and Widom \cite{Tracy} used the Fredholm theory of integral equations to study the probability distribution of the largest and smallest eigenvalue in the Gaussian unitary ensemble. They expressed it as a Fredholm determinant $\det(I-K)$, and proved that the logarithmic derivative of $\det(I-K)$ satisfies a third order differential equation. This equation can be integrated to the $\sigma$ form of a Painlev\'{e} IV if the integration constant is zero (see (5.14) in \cite{Tracy}). They also studied the largest and smallest eigenvalue distribution in the Laguerre and Jacobi unitary ensemble. Similarly, the Laguerre case is related to a third order differential equation and can be integrated to the $\sigma$ form of a Painlev\'{e} V. But for the Jacobi case, they only obtained a third order differential equation. Later, Haine and Semengue \cite{Haine} used the Virasoro approach to obtain another third order differential equation and showed that the difference of them can be reduced to the $\sigma$ form of a Painlev\'{e} VI.

The approach in this paper is based on the ladder operators of orthogonal polynomials and the associated compatibility conditions ($S_{1}$), ($S_{2}$) and ($S_{2}'$). An elementary method to compute $D_{n}[w]$ is to write it as a Hankel determinant. It is a well-known fact that Hankel determinants can be expressed as the product of the square of the $L^{2}$ norms of orthogonal polynomials. Based on the ladder operators adapted to these orthogonal polynomials, and from the associated supplementary conditions, a series of difference and differential equations can be derived to ultimately give a description of $D_{n}[w]$.

This paper is organized as follows: In Sec. 2, we consider the Gaussian weight with a single jump. From the ladder operators and supplementary conditions on the weight $w(x,t_{1})$, we obtain in addition to what was found in Chen and Pruessner \cite{Chen2005}, and prove that the auxiliary quantities $r_{n}(t_{1})$, $R_{n}(t_{1})$ and $\sigma_{n}(t_{1})$ satisfy the second order difference and differential equations respectively. Moreover, we consider the large $n$ behavior of the monic orthogonal polynomials $P_{n}(z,t_{1})$ and the double scaling of the auxiliary quantities. In Sec. 3, we consider the general case with two jumps. We also use the ladder operator approach to obtain a partial differential equation on $\sigma_{n}(t_{1},t_{2})$, which is a two variables' generalization of the $\sigma$ form of the Painlev\'{e} IV. The conclusion is given in Sec. 4.

\section{Gaussian Weight with a Single Jump}
We set $B_{2}=0$ and $B_{1}\neq 0$, which corresponds to the Gaussian weight with a single jump
\be\label{weight}
w(x,t_{1}):=\mathrm{e}^{-x^2}(A+B_{1}\theta(x-t_{1})).
\ee
This case has been studied by Chen and Pruessner \cite{Chen2005} with parameters $A=1-\frac{\beta}{2},\;B_{1}=\beta$. They obtained a Painlev\'{e} IV for the diagonal recurrence coefficient $\alpha_{n}(t_{1})$ of the monic orthogonal polynomials with respect to (\ref{weight}), which is actually the same with our result (\ref{p4}). Furthermore, Its \cite{Its} studied the weight
$$
w(x,t_{1})=\mathrm{e}^{-x^2}\cdot\left\{
\begin{aligned}
&\mathrm{e}^{i\beta \pi},&x<t_{1};\\
&\mathrm{e}^{-i\beta \pi},&x>t_{1},
\end{aligned}
\right.
$$
which corresponds to $A=\mathrm{e}^{i\beta \pi},\; B_{1}=\mathrm{e}^{-i\beta \pi}-\mathrm{e}^{i\beta \pi}$ in our problem, and they also obtained a Painlev\'{e} IV for $\alpha_{n}(t_{1})$ by using the Riemann-Hilbert approach. See also \cite{Bogatskiy, Its2008} for more detailed discussion on this weight. In addition, Xu and Zhao \cite{Xu} considered the Gaussian weight with a jump at the edge.

\subsection{Ladder Operators and Supplementary Conditions}
Let $P_{n}(x,t_{1})$ be the monic polynomials of degree $n$ orthogonal with respect to the weight function $w(x,t_{1})$, that is,
\be\label{or}
\int_{-\infty}^{\infty}P_{j}(x,t_{1})P_{k}(x,t_{1})w(x,t_{1})dx=h_{j}(t_{1})\delta_{jk},\;\;j,k=0,1,2,\ldots,
\ee
where
$$
w(x,t_{1})=w_{0}(x)(A+B_{1}\theta(x-t_{1})),\;\;w_{0}(x):=\mathrm{e}^{-\mathrm{v}_{0}(x)},\;\;\mathrm{v}_{0}(x)=x^2.
$$
The monic polynomials $P_{n}(x,t_{1})$ have the monomial expansion
\be\label{expan}
P_{n}(x,t_{1})=x^{n}+\mathrm{p}(n,t_{1})x^{n-1}+\cdots+P_{n}(0,t_{1}),
\ee
and we see later that $\mathrm{p}(n,t_{1})$, the coefficient of $x^{n-1}$, will play a significant role.

From the orthogonality condition (\ref{or}), we have the recurrence relation \cite{Szego}
\be\label{rr}
xP_{n}(x,t_{1})=P_{n+1}(x,t_{1})+\alpha_{n}(t_{1})P_{n}(x,t_{1})+\beta_{n}(t_{1})P_{n-1}(x,t_{1})
\ee
with the initial conditions
$$
P_{0}(x,t_{1})=1,\;\;\beta_{0}(t_{1})P_{-1}(x,t_{1})=0.
$$
An easy consequence of (\ref{rr}) and (\ref{expan}) gives
\be\label{alp}
\alpha_{n}(t_{1})=\mathrm{p}(n,t_{1})-\mathrm{p}(n+1,t_{1})
\ee
and a telescopic sum followed by
\be\label{sum}
\sum_{j=0}^{n-1}\alpha_{j}(t_{1})=-\mathrm{p}(n,t_{1}).
\ee
Moreover, it follows from (\ref{rr}) and (\ref{or}) that
\be\label{be}
\beta_{n}(t_{1})=\frac{h_{n}(t_{1})}{h_{n-1}(t_{1})}.
\ee
In this section, we also denote
$$
D_{n}(t_{1}):=D_{n}[w]=\frac{1}{n!}\int_{(-\infty,\infty)^{n}}\prod_{1\leq j<k\leq n}\left(x_{k}-x_{j}\right)^{2}
\prod_{j=1}^{n}w(x_{j},t_{1})dx_{j}.
$$
It is a well-known fact that $D_{n}(t_{1})$ can be expressed as a Hankel determinant generated by the weight $w(x,t_{1})$ \cite{Szego}
\bea\label{dnw}
D_{n}(t_{1})&=&\det\left(\int_{-\infty}^{\infty}x^{i+j}w(x,t_{1})dx\right)_{i,j=0}^{n-1}\nonumber\\
&=&\prod_{j=0}^{n-1}h_{j}(t_{1}).
\eea
For convenience, we would not show the $t_{1}$ dependence in $P_{n}(x)$, $h_{n}$, $\alpha_{n}$ and $\beta_{n}$ unless it is needed in the following discussions. From Lemma 1 in \cite{Basor2009} and noting that $\frac{\mathrm{v}_{0}'(z)-\mathrm{v}_{0}'(y)}{z-y}=2$ in our problem, we have the following theorem.
\begin{theorem}
The monic orthogonal polynomials with respect to the weight $w(x,t_{1})$ satisfy the following differential recurrence relations:
\be\label{lowering}
P_{n}'(z)=\beta_{n}A_{n}(z)P_{n-1}(z)-B_{n}(z)P_{n}(z),
\ee
\be\label{raising}
P_{n-1}'(z)=(B_{n}(z)+\mathrm{v}_{0}'(z))P_{n-1}(z)-A_{n-1}(z)P_{n}(z),
\ee
where
\be\label{anz}
A_{n}(z)=2+\frac{R_{n}(t_{1})}{z-t_{1}},
\ee
\be\label{bnz}
B_{n}(z)=\frac{r_{n}(t_{1})}{z-t_{1}},
\ee
and
$$
R_{n}(t_{1})=\frac{B_{1}\: P_{n}^{2}(t_{1},t_{1})\mathrm{e}^{-t_{1}^{2}}}{h_{n}(t_{1})},\;\;\;\;\;\;r_{n}(t_{1})=\frac{B_{1}\: P_{n}(t_{1},t_{1})P_{n-1}(t_{1},t_{1})\mathrm{e}^{-t_{1}^{2}}}{h_{n-1}(t_{1})}.
$$
Here $P_{n}(t_{1},t_{1}):=P_{n}(x,t_{1})|_{x=t_{1}}$.
\end{theorem}
The following lemmas can be found in \cite{Basor2009, Chen2005, Chen2010}. See also \cite{Chen1997,Dai,Magnus,Min2017} for more information.
\begin{lemma}\label{s1s2}
The functions $A_{n}(z)$ and $B_{n}(z)$ satisfy the conditions:
\be
B_{n+1}(z)+B_{n}(z)=(z-\alpha_{n}) A_{n}(z)-\mathrm{v}_{0}'(z), \tag{$S_{1}$}
\ee
\be
1+(z-\alpha_{n})(B_{n+1}(z)-B_{n}(z))=\beta_{n+1}A_{n+1}(z)-\beta_{n}A_{n-1}(z). \tag{$S_{2}$}
\ee
\end{lemma}
The combination of ($S_{1}$) and ($S_{2}$) produces the following identity which gives better insight into the $\beta_{n}$ term.
\begin{lemma}\label{s2p}
$A_{n}(z)$, $B_{n}(z)$ and $\sum_{j=0}^{n-1}A_{j}(z)$ satisfy the identity
\be
B_{n}^{2}(z)+\mathrm{v}_{0}'(z)B_{n}(z)+\sum_{j=0}^{n-1}A_{j}(z)=\beta_{n}A_{n}(z)A_{n-1}(z). \tag{$S_{2}'$}
\ee
\end{lemma}
\noindent $\mathbf{Remark.}$ The three identities ($S_{1}$), ($S_{2}$) and ($S_{2}'$) are valid for $z\in \mathbb{C}\cup\{\infty\}$.
\begin{lemma}\label{de}
$P_{n}(z)$ satisfy the following second order differential equation:
\be\label{sode}
P_{n}''(z)-\left(\mathrm{v}_{0}'(z)+\frac{A_{n}'(z)}{A_{n}(z)}\right)P_{n}'(z)+\left(B_{n}'(z)-B_{n}(z)\frac{A_{n}'(z)}{A_{n}(z)}
+\sum_{j=0}^{n-1}A_{j}(z)\right)P_{n}(z)=0.
\ee
\end{lemma}
Substituting (\ref{anz}) and (\ref{bnz}) into ($S_{1}$), we obtain
\be\label{s1}
\frac{r_{n+1}(t_{1})+r_{n}(t_{1})}{z-t_{1}}=\frac{(z-\alpha_{n})R_{n}(t_{1})}{z-t_{1}}-2\alpha_{n}.
\ee
Letting $z\rightarrow\infty$, we have
\be\label{s11}
R_{n}(t_{1})=2\alpha_{n}.
\ee
Equating the residues at the simple pole $t$ from (\ref{s1}), we find
\be\label{s12}
r_{n+1}(t_{1})+r_{n}(t_{1})=(t_{1}-\alpha_{n})R_{n}(t_{1}).
\ee

Similarly, plugging (\ref{anz}) and (\ref{bnz}) into ($S_{2}'$), we obtain
\be\label{s2}
\frac{r_{n}^{2}(t_{1})}{(z-t_{1})^2}+\frac{2z r_{n}(t_{1})+\sum_{j=0}^{n-1}R_{j}(t_{1})}{z-t_{1}}+2n=\frac{\beta_{n}R_{n}(t_{1})R_{n-1}(t_{1})}{(z-t_{1})^2}
+\frac{2\beta_{n}(R_{n}(t_{1})+R_{n-1}(t_{1}))}{z-t_{1}}+4\beta_{n}.
\ee
Letting $z\rightarrow\infty$, we have
\be\label{s23}
r_{n}(t_{1})=2\beta_{n}-n.
\ee
Multiplying both sides of (\ref{s2}) by $(z-t_{1})^2$ and letting $z\rightarrow\infty$ gives
\be\label{s21}
r_{n}^{2}(t_{1})=\beta_{n}R_{n}(t_{1})R_{n-1}(t_{1}).
\ee
Then (\ref{s2}) becomes
$$
\frac{2z r_{n}(t_{1})+\sum_{j=0}^{n-1}R_{j}(t_{1})}{z-t_{1}}+2n=\frac{2\beta_{n}(R_{n}(t_{1})+R_{n-1}(t_{1}))}{z-t_{1}}+4\beta_{n}.
$$
Equating the residues at the simple pole $t_{1}$ produces
\be\label{s22a}
2t_{1}r_{n}(t_{1})-2\beta_{n}(R_{n}(t_{1})+R_{n-1}(t_{1}))+\sum_{j=0}^{n-1}R_{j}(t_{1})=0.
\ee
Using (\ref{s23}) and (\ref{s21}) to eliminate $\beta_{n}$ and $R_{n-1}(t_{1})$ from the above, we have
\be\label{s22}
2t_{1}r_{n}(t_{1})-(n+r_{n}(t_{1}))R_{n}(t_{1})-\frac{2r_{n}^2(t_{1})}{R_{n}(t_{1})}+\sum_{j=0}^{n-1}R_{j}(t_{1})=0.
\ee
$\mathbf{Remark.}$ The above equalities (\ref{s11}), (\ref{s12}), (\ref{s23}), (\ref{s21}) and (\ref{s22}) also appeared in \cite{Chen2005}, and are crucial for the following discussions.

\subsection{Non-linear Difference Equations Satisfied by $r_{n}(t_{1}),\;R_{n}(t_{1})$ and $\sigma_{n}(t_{1})$}
In this subsection, we would like to obtain the second order non-linear difference equations satisfied by $r_{n}(t_{1}),\;R_{n}(t_{1})$ and $\sigma_{n}(t_{1})$ respectively. These are the new results beyond \cite{Chen2005}.

Eliminating $\alpha_{n}(t_{1})$ from (\ref{s11}) and (\ref{s12}), we have
\be\label{s121}
r_{n+1}(t_{1})+r_{n}(t_{1})=\left(t_{1}-\frac{1}{2}R_{n}(t_{1})\right)R_{n}(t_{1}).
\ee
Similarly, eliminating $\beta_{n}(t_{1})$ from (\ref{s23}) and (\ref{s21}), we find
\be\label{s211}
r_{n}^{2}(t_{1})=\frac{1}{2}(n+r_{n}(t_{1}))R_{n}(t_{1})R_{n-1}(t_{1}).
\ee
Solving for $R_{n}(t_{1})$ from (\ref{s121}), a quadratic, and substituting either solution into (\ref{s211}), we find after clearing the square root, a non-linear second order difference equation for $r_{n}:=r_{n}(t_{1})$,
\bea\label{dr}
&&\left[(n+r_n)^2(r_{n-1}+r_n-t_{1}^2)r_{n+1}+(n+r_n)^2r_n r_{n-1}+2 n\: r_n^3-n (t_{1}^2-n) r_n^2-n^2 t_{1}^2\: r_n\right]^2\nonumber\\
&=&t_{1}^2 (n+r_n)^2(t_{1}^2-2 r_{n-1}-2 r_n)  \left[n\: r_n+(n+r_n) r_{n+1}\right]^2.
\eea

On the other hand, regarding (\ref{s211}) as a quadratic equation in $r_{n}(t_{1})$, and substituting either solution into (\ref{s121}), we find after clearing the square root, the following second order non-linear difference equation for $R_{n}:=R_{n}(t_{1})$,
\bea\label{dr1}
&&\left[(2 R_n^2-4 t_{1} R_n+R_{n-1}R_n-4n-4)R_{n+1}+2(R_n^2-2 t_{1} R_n-2 n)R_{n-1}+2 R_n(R_n-2 t_{1})^2\right]^2\nonumber\\
&=&R_{n-1}R_{n+1}(R_{n-1}R_n+8 n)(R_nR_{n+1}+8 n+8).
\eea
$\mathbf{Remark.}$ (\ref{dr1}) may be related to the discrete Painlev\'{e} IV equation \cite{Forrester, Ramani1991, Ramani1996}.

Now we introduce a quantity $\sigma_{n}(t_{1})$, defined by
\be\label{sig}
\sigma_{n}(t_{1}):=-\sum_{j=0}^{n-1}R_{j}(t_{1}).
\ee
We will see in the next subsection that $\sigma_{n}(t_{1})$ is the logarithmic derivative of $D_{n}(t_{1})$, i.e.
$$
\sigma_{n}(t_{1})=\frac{d}{dt_{1}}\ln D_{n}(t_{1}).
$$
It is easy to see that
\be\label{rs}
R_{n}(t_{1})=\sigma_{n}(t_{1})-\sigma_{n+1}(t_{1}).
\ee
Using (\ref{s23}), (\ref{sig}) and (\ref{rs}), (\ref{s22a}) becomes
$$
2t_{1}r_{n}(t_{1})-(n+r_{n}(t_{1}))(\sigma_{n-1}(t_{1})-\sigma_{n+1}(t_{1}))-\sigma_{n}(t_{1})=0.
$$
Then,
\be\label{rnt1}
r_{n}(t_{1})=\frac{n(\sigma_{n-1}(t_{1})-\sigma_{n+1}(t_{1}))+\sigma_{n}(t_{1})}{2t_{1}-\sigma_{n-1}(t_{1})+\sigma_{n+1}(t_{1})}.
\ee
With the aid of (\ref{s23}) and (\ref{rs}), it follows from (\ref{s21}) that
\be\label{rnts}
r_{n}^{2}(t_{1})=\frac{n+r_{n}(t_{1})}{2}(\sigma_{n}(t_{1})-\sigma_{n+1}(t_{1}))(\sigma_{n-1}(t_{1})-\sigma_{n}(t_{1})).
\ee
Substituting (\ref{rnt1}) into (\ref{rnts}), we obtain a second order difference equation satisfied by $\sigma_{n}:=\sigma_{n}(t_{1})$,
$$
2\left[\sigma_{n}+n(\sigma_{n-1}-\sigma_{n+1})\right]^2=
(\sigma_{n}-\sigma_{n+1})(\sigma_{n-1}-\sigma_{n})(\sigma_{n}+2nt_{1})(\sigma_{n+1}
-\sigma_{n-1}+2t_{1}),
$$
which is the discrete $\sigma$ form of the Painlev\'{e} IV equation.

\subsection{Painlev\'{e} IV, Chazy II, and the $\sigma$ Form}
In this subsection, we derive the second order differential equations satisfied by $R_{n}(t_{1})$ and $r_{n}(t_{1})$ respectively, which are related to the Painlev\'{e} IV and Chazy II. The Painlev\'{e} IV satisfied by $R_{n}(t_{1})$ or $\alpha_{n}(t_{1})$ was obtained by \cite{Chen2005}, but the Chazy equation satisfied by $r_{n}(t_{1})$ is a new one. We would like to say that Chazy equations appear in the random matrix theory regularly. For example, Witte, Forrester and Cosgrove \cite{Witte} obtained the Chazy equations when they studied the gap probability of Gaussian and Jacobi unitary ensembles. Recently, Lyu, Chen and Fan \cite{Lyu} also obtained the Chazy equation when they studied the gap probability of Gaussian unitary ensembles by using the different method. In the end, we find that $\sigma_{n}(t_{1})$, the logarithmic derivative of $D_{n}(t_{1})$, satisfies the Jimbo-Miwa-Okamoto $\sigma$ form of the Painlev\'{e} IV.

We begin with taking a derivative with respect to $t_{1}$ in the following equation,
$$
\int_{-\infty}^{\infty}P_{n}^2(x,t_{1})w(x,t_{1})dx=h_{n}(t_{1}),\;\;n=0,1,2,\ldots,
$$
we obtain
$$
h_{n}'(t_{1})=-B\:P_{n}^{2}(t_{1})\mathrm{e}^{-t_{1}^2}.
$$
It follows that
\be\label{hn}
(\ln h_{n}(t_{1}))'=-R_{n}(t_{1})
\ee
and
$$
[\ln\beta_{n}(t_{1})]'=(\ln h_{n}(t_{1}))'-(\ln h_{n-1}(t_{1}))'=R_{n-1}(t_{1})-R_{n}(t_{1}),
$$
where we have made use of (\ref{be}) in the first step.
Hence,
\be\label{e1}
\beta_{n}'(t_{1})=\beta_{n}R_{n-1}(t_{1})-\beta_{n}R_{n}(t_{1}).
\ee
Moreover, from (\ref{dnw}), (\ref{hn}) and (\ref{sig}) we obtain
\be\label{sigma}
\frac{d}{dt_{1}}\ln D_{n}(t_{1})=\frac{d}{dt_{1}}\ln\prod_{j=0}^{n-1}h_{j}(t_{1})=-\sum_{j=0}^{n-1}R_{j}(t_{1})=\sigma_{n}(t_{1}).
\ee
On the other hand, differentiating with respect to $t_{1}$ in the equation
$$
\int_{-\infty}^{\infty}P_{n}(x,t_{1})P_{n-1}(x,t_{1})w(x,t_{1})dx=0,\;\;n=0,1,2,\ldots,
$$
gives
\be\label{dp}
\frac{d}{dt_{1}}\mathrm{p}(n,t_{1})=r_{n}(t_{1}).
\ee
Now we have the Toda equations on $\alpha_{n}(t_{1})$ and $\beta_{n}(t_{1})$. These are also obtained by \cite{Chen2005}.
\begin{proposition}
\be\label{toda1}
\beta_{n}'(t_{1})=2\beta_{n}(\alpha_{n-1}-\alpha_{n})
\ee
\be\label{toda2}
\alpha_{n}'(t_{1})=2(\beta_{n}-\beta_{n+1})+1
\ee
\end{proposition}
\begin{proof}
The combination of (\ref{s11}) and (\ref{e1}) results in (\ref{toda1}). (\ref{toda2}) comes from (\ref{alp}), (\ref{s23}) and (\ref{dp}).
\end{proof}
The following lemma is very important for the derivation of the second order differential equations satisfied by $R_{n}(t_{1})$, $r_{n}(t_{1})$ and $\sigma_{n}(t_{1})$.
\begin{lemma}
$r_{n}(t_{1})$ and $R_{n}(t_{1})$ satisfy the following coupled Riccati equations:
\be\label{ri1}
r_{n}'(t_{1})=\frac{2r_{n}^2(t_{1})}{R_{n}(t_{1})}-(n+r_{n}(t_{1}))R_{n}(t_{1}),
\ee
\be\label{ri2}
R_{n}'(t_{1})=R_{n}^2(t_{1})-2t_{1}R_{n}(t_{1})+4r_{n}(t_{1}).
\ee
\end{lemma}

\begin{proof}
From (\ref{e1}) and (\ref{s21}), we have
\be\label{bnpt}
\beta_{n}'(t_{1})=\frac{r_{n}^2(t_{1})}{R_{n}(t_{1})}-\beta_{n}R_{n}(t_{1}).
\ee
It follows from (\ref{s23}) that
\be\label{bn}
\beta_{n}=\frac{n+r_{n}(t_{1})}{2},
\ee
Substituting (\ref{bn}) into (\ref{bnpt}), we obtain the Riccati equation satisfied by $r_{n}(t_{1})$,
$$
r_{n}'(t_{1})=\frac{2r_{n}^2(t_{1})}{R_{n}(t_{1})}-(n+r_{n}(t_{1}))R_{n}(t_{1}).
$$

Now we come to prove the second equation (\ref{ri2}), the Riccati equation satisfied by $R_{n}(t_{1})$. From (\ref{s11}) and (\ref{alp}) we find
$$
R_{n}(t_{1})=2\left[\mathrm{p}(n,t_{1})-\mathrm{p}(n+1,t_{1})\right],
$$
then with the aid of (\ref{dp}),
\bea
R_{n}'(t_{1})&=&2\left[\frac{d\:\mathrm{p}(n,t_{1})}{dt_{1}}-\frac{d\:\mathrm{p}(n+1,t_{1})}{dt_{1}}\right]\nonumber\\
&=&2\left[r_{n}(t_{1})-r_{n+1}(t_{1})\right].\nonumber
\eea
Using (\ref{s12}) to decrease the index $n+1$ to $n$, we have
\bea
R_{n}'(t_{1})
&=&2\left[2r_{n}(t_{1})-(t_{1}-\alpha_{n})R_{n}(t_{1})\right]\nonumber\\
&=&2\left[2r_{n}(t_{1})-\left(t_{1}-\frac{1}{2}R_{n}(t_{1})\right)R_{n}(t_{1})\right]\nonumber\\
&=&R_{n}^2(t_{1})-2t_{1}R_{n}(t_{1})+4r_{n}(t_{1}),\nonumber
\eea
where we have used (\ref{s11}) in the second equality. This establishes the lemma.
\end{proof}

\begin{theorem}\label{pc}
$R_{n}(t_{1})$ and $r_{n}(t_{1})$ satisfy the following second order differential equations,
\be\label{sod}
R_{n}''(t_{1})=\frac{(R_{n}'(t_{1}))^{2}}{2R_{n}(t_{1})}+\frac{3}{2}R_{n}^3(t_{1})-4t_{1}R_{n}^2(t_{1})+2(t_{1}^2-2n-1)R_{n}(t_{1})
\ee
and
$$
\left[r_{n}''(t_{1})+12r_{n}^2(t_{1})+8n\:r_{n}(t_{1})\right]^2=4t_{1}^2\left[(r_{n}'(t_{1}))^2+8r_{n}^3(t_{1})+8n\:r_{n}^2(t_{1})\right]
$$
respectively.
Moreover, letting $y(t_{1})=R_{n}(-t_{1})$, then $y(t_{1})$ satisfies the Painlev\'{e} IV equation \cite{Gromak},
\be\label{p4}
y''(t_{1})=\frac{(y'(t_{1}))^{2}}{2y(t_{1})}+\frac{3}{2}y^3(t_{1})+4t_{1}y^2(t_{1})+2(t_{1}^2-\alpha_{1})y(t_{1})+\frac{\beta_{1}}{y(t_{1})}
\ee
with $\alpha_{1}=2n+1,\;\beta_{1}=0$. Letting $v(t_{1})=-2r_{n}(t_{1})-\frac{2n}{3}$, then $v(t_{1})$ satisfies the first member of the Chazy II system \cite{Cosgrove},
$$
\left(v''(t_{1})-6v^2(t_{1})-\alpha_{2}\right)^2=4t_{1}^2(v''(t_{1})-4v^3(t_{1})-2\alpha_{2}\:v(t_{1})-\beta_{2})
$$
with $\alpha_{2}=-\frac{8n^2}{3},\;\beta_{2}=-\frac{64n^3}{27}$.
\end{theorem}

\begin{proof}
We start from expressing $r_{n}(t_{1})$ in terms of $R_{n}(t_{1})$ and $R_{n}'(t_{1})$ by (\ref{ri2}),
\be\label{rn}
r_{n}(t_{1})=\frac{1}{4}\left[R_{n}'(t_{1})-R_{n}^2(t_{1})+2t_{1}R_{n}(t_{1})\right],
\ee
then
\be\label{rnp}
r_{n}'(t_{1})=\frac{1}{4}\left[R_{n}''(t_{1})-2R_{n}(t_{1})R_{n}'(t_{1})+2R_{n}(t_{1})+2tR_{n}'(t_{1})\right].
\ee
Substituting (\ref{rn}) and (\ref{rnp}) into (\ref{ri1}), we obtain the differential equation for $R_{n}(t_{1})$,
$$
R_{n}''(t_{1})=\frac{(R_{n}'(t_{1}))^{2}}{2R_{n}(t_{1})}+\frac{3}{2}R_{n}^3(t_{1})-4t_{1}R_{n}^2(t_{1})+2(t_{1}^2-2n-1)R_{n}(t_{1}).
$$
Letting $y(t_{1})=R_{n}(-t_{1})$, it is easy to see that $y(t_{1})$ satisfies a particular Painlev\'{e} IV,
$$
y''(t_{1})=\frac{(y'(t_{1}))^{2}}{2y(t_{1})}+\frac{3}{2}y^3(t_{1})+4t_{1}y^2(t_{1})+2(t_{1}^2-2n-1)y(t_{1}).
$$

Now we turn to prove the differential equation for $r_{n}(t_{1})$. Viewing (\ref{ri1}) as an equation on $R_{n}(t_{1})$, we find the solution is
$$
R_{n}(t_{1})=\frac{-r_{n}'(t_{1})\pm\sqrt{\Delta}}{2(n+r_{n}(t_{1}))},
$$
where
$$
\Delta:=(r_{n}'(t_{1}))^2+8r_{n}^3(t_{1})+8nr_{n}^2(t_{1}).
$$
Then plugging it into (\ref{ri2}), we have
$$
\frac{\left(-r_{n}'(t_{1})\pm\sqrt{\Delta}\right)\left(-r_{n}''(t_{1})-12r_{n}^2(t_{1})-8nr_{n}(t_{1})\pm 2t_{1}\sqrt{\Delta}\right)}
{2(n+r_{n}(t_{1}))\sqrt{\Delta}}=0.
$$
It follows a second order differential equation for $r_{n}(t_{1})$,
$$
-r_{n}''(t_{1})-12r_{n}^2(t_{1})-8nr_{n}(t_{1})\pm 2t_{1}\sqrt{(r_{n}'(t_{1}))^2+8r_{n}^3(t_{1})+8nr_{n}^2(t_{1})}=0,
$$
or
\be\label{difr}
\left[r_{n}''(t_{1})+12r_{n}^2(t_{1})+8nr_{n}(t_{1})\right]^2=4t_{1}^2\left[(r_{n}'(t_{1}))^2+8r_{n}^3(t_{1})+8nr_{n}^2(t_{1})\right].
\ee
Letting $v(t_{1})=-2r_{n}(t_{1})-\frac{2n}{3}$, or $r_{n}(t_{1})=-\frac{v(t_{1})}{2}-\frac{n}{3}$, and substituting it into (\ref{difr}), then $v(t_{1})$ satisfies the first member of the Chazy II system,
$$
\left(v''(t_{1})-6v^2(t_{1})+\frac{8n^2}{3}\right)^2=4t_{1}^2\left(v''(t_{1})-4v^3(t_{1})+\frac{16n^2}{3}v(t_{1})+\frac{64n^3}{27}\right).
$$
This finishes the proof of Theorem \ref{pc}.
\end{proof}

\begin{theorem}
$\sigma_{n}(t_{1})$ satisfies the Jimbo-Miwa-Okamoto $\sigma$ form of the Painlev\'{e} IV equation \cite{Jimbo1981}
\be\label{jmo}
(\sigma_{n}''(t_{1}))^2=4\left(t_{1}\sigma_{n}'(t_{1})-\sigma_{n}(t_{1})\right)^2
-4(\sigma_{n}'(t_{1})+\nu_{0})(\sigma_{n}'(t_{1})+\nu_{1})(\sigma_{n}'(t_{1})+\nu_{2}),
\ee
with parameters $\nu_{0}=\nu_{1}=0,\: \nu_{2}=2n$.
\end{theorem}
\begin{proof}
By using (\ref{sig}), (\ref{s22}) becomes
\be\label{eq1}
(n+r_{n}(t_{1}))R_{n}(t_{1})+\frac{2r_{n}^2(t_{1})}{R_{n}(t_{1})}=2tr_{n}(t_{1})-\sigma_{n}(t_{1}).
\ee
From (\ref{ri1}), we have
\be\label{eq2}
(n+r_{n}(t_{1}))R_{n}(t_{1})-\frac{2r_{n}^2(t_{1})}{R_{n}(t_{1})}=-r_{n}'(t_{1}).
\ee
The difference and sum of (\ref{eq1}) and (\ref{eq2}) give
\be\label{eq3}
\frac{4r_{n}^2(t_{1})}{R_{n}(t_{1})}=2t_{1}r_{n}(t_{1})-\sigma_{n}(t_{1})+r_{n}'(t_{1})
\ee
and
\be\label{eq4}
2(n+r_{n}(t_{1}))R_{n}(t_{1})=2t_{1}r_{n}(t_{1})-\sigma_{n}(t_{1})-r_{n}'(t_{1}),
\ee
respectively.
\\
Then the product of (\ref{eq3}) and (\ref{eq4}) leads to
\be\label{eq5}
8(n+r_{n}(t_{1}))r_{n}^2(t_{1})=(2t_{1}r_{n}(t_{1})-\sigma_{n}(t_{1}))^2-(r_{n}'(t_{1}))^2.
\ee
Noting that from (\ref{sig}), (\ref{s11}) and (\ref{sum}), we find
$$
\sigma_{n}(t_{1})=-2\sum_{j=0}^{n-1}\alpha_{j}=2\mathrm{p}(n,t_{1}).
$$
Then it follows from (\ref{dp}) that
$$
\sigma_{n}'(t_{1})=2\frac{d}{dt_{1}}\mathrm{p}(n,t_{1})=2r_{n}(t_{1}),
$$
or
\be\label{rnt}
r_{n}(t_{1})=\frac{1}{2}\sigma_{n}'(t_{1}).
\ee
Substituting (\ref{rnt}) into (\ref{eq5}), we obtain a second order differential equation satisfied by $\sigma_{n}(t_{1})$,
$$
(\sigma_{n}''(t_{1}))^2=4\left(t_{1}\sigma_{n}'(t_{1})-\sigma_{n}(t_{1})\right)^2-4(\sigma_{n}'(t_{1}))^2(\sigma_{n}'(t_{1})+2n),
$$
which is just the Jimbo-Miwa-Okamoto $\sigma$ form of the Painlev\'{e} IV equation, $\mathrm{P_{IV}}(0,0,2n)$.
\end{proof}
This result is coincident with Tracy and Widom \cite{Tracy} when they studied the largest eigenvalue distribution in the Gaussian unitary ensemble. Therefore, $\sigma_{n}(t_{1})$ satisfies both the continuous and discrete $\sigma$ form of the Painlev\'{e} IV.

\subsection{Large $n$ Behavior of the Orthogonal Polynomials and Double Scaling Analysis}
In this subsection, we consider the large $n$ behavior of the monic orthogonal polynomials $P_{n}(z)$. We show that, as $n\rightarrow\infty$, $P_{n}(z)$ satisfies the confluent forms of Heun's differential equation. We also give the large $n$ asymptotics of $R_{n}(t_{1}), r_{n}(t_{1})$ and $\sigma_{n}(t_{1})$ under a double scaling, which gives rise to the Painlev\'{e} XXXIV equation.
\begin{theorem}
As $n\rightarrow\infty$, $\mathrm{(i)}$ if $B_{1}>0$, then $\hat{P}_{n}(u):=P_{n}\left(\frac{u}{\sqrt{2}}+t_{1}\right)$ satisfies the biconfluent Heun equation (BHE) \cite{Ronveaux}
\be\label{bhe}
\hat{P}_{n}''(u)-\left(\frac{\gamma}{u}+\delta+u\right)\hat{P}_{n}'(u)+\frac{\alpha\: u-q}{u}\hat{P}_{n}(u)=0,
\ee
where $\gamma=-1,\; \delta=\sqrt{2}t_{1},\; \alpha=0,\; q=-\frac{4\sqrt{3}n^{\frac{3}{2}}}{9}$;\\
$\mathrm{(ii)}$ if $B_{1}<0$, then $\hat{P}_{n}(u):=P_{n}\left(\frac{u}{\sqrt{2}}+t_{1}\right)$ satisfies (\ref{bhe}) with parameters $\gamma=-1,\; \delta=\sqrt{2}t_{1},\; \alpha=0,\; q=\frac{4\sqrt{3}n^{\frac{3}{2}}}{9}$.
\end{theorem}
\begin{proof}
Substituting (\ref{anz}) and (\ref{bnz}) into (\ref{sode}), and using (\ref{s22}) to eliminate $\sum_{j=0}^{n-1}R_{j}(t_{1})$, we obtain
\bea\label{ode}
&&P_n''(z)+P_n'(z) \left(\frac{R_n(t_1)}{(z-t_1) (2z-2 t_1+R_n(t_1))}-2z\right)+P_n(z) \bigg(2n-\frac{r_n(t_1)}{(z-t_1)^2}\nonumber\\
&+&\frac{r_n(t_1)R_n(t_1)}{(z-t_1)^2 (2z-2 t_1+R_n(t_1))}+\frac{2 r_n^2(t_1)+(n+r_n(t_1)) R_n^2(t_1)-2 t_1 r_n(t_1)R_n(t_1)}{(z-t_1)R_n(t_1)}\bigg)=0.\nonumber\\
\eea
Replacing $r_n(t_1)$ with the expression of $R_n(t_1)$ from (\ref{rn}), (\ref{ode}) becomes
\bea\label{ode1}
&&P_n''(z)+P_n'(z) \left(\frac{R_n(t_1)}{(z-t_1) (2z-2 t_1+R_n(t_1))}-2z\right)+P_n(z)\bigg(2 n-\frac{R_n'(t_1)-R_n^2(t_1)+2 t_1 R_n(t_1)}{4 (z-t_1)^2}\nonumber\\[10pt]
&+&\frac{R_n(t_1) \left(R_n'(t_1)-R_n^2(t_1)+2 t_1 R_n(t_1)\right)}{4 (z-t_1)^2 (2z-2 t_1+R_n(t_1))}+\frac{(R_n'(t_1))^2-R_n^4(t_1)+4 t_1 R_n^3(t_1)+(8 n-4 t_1^2) R_n^2(t_1)}{8(z-t_1) R_n(t_1)}\bigg)=0.\nonumber\\
\eea
Note that the coefficients of $P_n(z)$ and $P_n'(z)$ only depend on $R_n(t_1)$ and $R_n'(t_1)$ in (\ref{ode1}). Now we consider the large $n$ behavior of $R_n(t_1)$. Let $\hat{R}_n(t_1)$ satisfy a quadratic equation obtained from the non-derivative part of (\ref{sod}),
$$
3\hat{R}_{n}^2(t_{1})-8t_{1}\hat{R}_{n}(t_{1})+4(t_{1}^2-2n-1)=0
$$
with the solutions
$$
\hat{R}_{n}(t_{1})=\frac{2}{3} \left(2 t_{1}\pm\sqrt{t_{1}^2+6n+3}\right).
$$

(i) If $B_{1}>0$, we choose
$$
\hat{R}_{n}(t_{1})=\frac{2}{3} \left(2 t_{1}+\sqrt{t_{1}^2+6n+3}\right).
$$
As $n\rightarrow\infty$,
$$
\hat{R}_{n}(t_{1})=\frac{2\sqrt{6n}}{3}+\frac{4 t_{1}}{3}+\frac{\sqrt{6}(t_{1}^2+3)}{18\sqrt{n}}-\frac{\sqrt{6}(t_{1}^2+3)^2}{432n^{\frac{3}{2}}}
+\frac{\sqrt{6}(t_{1}^2+3)^3}{5184n^{\frac{5}{2}}}+O(n^{-\frac{7}{2}}).
$$
Hence we suppose the expansion, as $n\rightarrow\infty$,
$$
R_{n}(t_{1})=\sum_{j=0}^{\infty}a_{j}(t_1)n^{\frac{1-j}{2}}.
$$
Substituting the above into (\ref{sod}), we obtain
\be\label{ex1}
R_{n}(t_{1})=\frac{2\sqrt{6n}}{3}+\frac{4 t_{1}}{3}+\frac{\sqrt{6}(t_{1}^2+3)}{18\sqrt{n}}-\frac{\sqrt{6}( t_{1}^4+6  t_{1}^2+15) }{432n^{\frac{3}{2}}}+\frac{t_{1}}{18n^2}+\frac{\sqrt{6} (t_{1}^6+9  t_{1}^4-117  t_{1}^2+81) }{5184n^{\frac{5}{2}}}+O(n^{-3}).
\ee
Plugging (\ref{ex1}) into (\ref{ode1}), we see that as $n\rightarrow\infty$,
$$
P_{n}''(z)+\left(\frac{1}{z-t_1}-2 z\right)P_{n}'(z)+\frac{4 \sqrt{6} n^{\frac{3}{2}}}{9 (z-t_1)}P_{n}(z)=0.
$$
Let
$$
z=\frac{u}{\sqrt{2}}+t_{1}.
$$
Then $\hat{P}_{n}(u):=P_{n}\left(\frac{u}{\sqrt{2}}+t_{1}\right)$
satisfies the biconfluent Heun equation (\ref{bhe}) with parameters $\gamma=-1,\; \delta=\sqrt{2}t_{1},\; \alpha=0,\; q=-\frac{4\sqrt{3}n^{\frac{3}{2}}}{9}$.

(ii) If $B_{1}<0$, we choose
$$
\hat{R}_{n}(t_{1})=\frac{2}{3} \left(2 t_{1}-\sqrt{t_{1}^2+6n+3}\right).
$$
As $n\rightarrow\infty$,
$$
\hat{R}_{n}(t_{1})=-\frac{2\sqrt{6n}}{3}+\frac{4 t_{1}}{3}-\frac{\sqrt{6}(t_{1}^2+3)}{18\sqrt{n}}+\frac{\sqrt{6}(t_{1}^2+3)^2}{432n^{\frac{3}{2}}}
-\frac{\sqrt{6}(t_{1}^2+3)^3}{5184n^{\frac{5}{2}}}+O(n^{-\frac{7}{2}}).
$$
Similarly, we suppose that as $n\rightarrow\infty$,
$$
R_{n}(t_{1})=\sum_{j=0}^{\infty}b_{j}(t_1)n^{\frac{1-j}{2}}.
$$
Substituting it into (\ref{sod}), we obtain
\be\label{ex2}
R_{n}(t_{1})=-\frac{2\sqrt{6n}}{3}+\frac{4 t_{1}}{3}-\frac{\sqrt{6}(t_{1}^2+3)}{18\sqrt{n}}+\frac{\sqrt{6}( t_{1}^4+6  t_{1}^2+15) }{432n^{\frac{3}{2}}}+\frac{t_{1}}{18n^2}-\frac{\sqrt{6} (t_{1}^6+9  t_{1}^4-117  t_{1}^2+81) }{5184n^{\frac{5}{2}}}+O(n^{-3}).
\ee
Plugging (\ref{ex2}) into (\ref{ode1}), we see that as $n\rightarrow\infty$,
$$
P_{n}''(z)+\left(\frac{1}{z-t_1}-2 z\right)P_{n}'(z)-\frac{4 \sqrt{6} n^{\frac{3}{2}}}{9 (z-t_1)}P_{n}(z)=0.
$$
Let
$$
z=\frac{u}{\sqrt{2}}+t_{1}.
$$
Then $\hat{P}_{n}(u):=P_{n}\left(\frac{u}{\sqrt{2}}+t_{1}\right)$
satisfies the biconfluent Heun equation ($\ref{bhe}$) with parameters $\gamma=-1,\; \delta=\sqrt{2}t_{1},\; \alpha=0,\; q=\frac{4\sqrt{3}n^{\frac{3}{2}}}{9}$.
\end{proof}
\noindent $\mathbf{Remark.}$ There are four standard confluent forms of Heun's equation: confluent Heun equation (CHE), doubly confluent Heun equation (DCHE), biconfluent Heun equation (BHE) and triconfluent Heun equation (THE) \cite{Ronveaux}. The Heun's equation and its confluent forms play an important role in mathematical physics. Many known special functions, such as hypergeometric functions, Mathieu functions and spheroidal functions are solutions of Heun-class equations \cite{Slavyanov}.

\begin{theorem}\label{cor}
Assume that $n\rightarrow\infty,\; t_{1}=\sqrt{2n}+2^{-\frac{1}{2}}n^{-\frac{1}{6}}s$ and $s$ is fixed. Then the large $n$ asymptotics of $R_{n}(t_{1}), r_{n}(t_{1})$ and $\sigma_{n}(t_{1})$ are given by
$$
R_{n}(t_{1})=n^{-\frac{1}{6}}v_{1}(s)+n^{-\frac{1}{2}}v_{2}(s)+n^{-\frac{5}{6}}v_{3}(s)+O(n^{-\frac{7}{6}}),
$$
\be\label{r1}
r_{n}(t_{1})=\frac{n^{\frac{1}{3}}v_{1}(s)}{\sqrt{2}}+\frac{\sqrt{2}}{4} \left( v_{1}'(s)+2 v_{2}(s)\right)+\frac{\sqrt{2} s\: v_{1}(s)-v_{1}^2(s)+\sqrt{2} v_{2}'(s)+2 \sqrt{2} v_{3}(s)}{4 n^{\frac{1}{3}}}+O(n^{-\frac{2}{3}}),
\ee
and
\bea\label{sigma1}
\sigma_{n}(t_{1})&=&n^{\frac{1}{6}}\left(s\: v_{1}(s)-\frac{v_{1}^2(s)}{\sqrt{2}}-\frac{(v_{1}'(s))^2}{4 v_{1}(s)}\right)+n^{-\frac{1}{6}}\left(s\: v_{2}(s)-\sqrt{2} v_{1}(s) v_{2}(s)-\frac{ v_{1}'(s) v_{2}'(s)}{2 v_{1}(s)}+\frac{ (v_{1}'(s))^2 v_{2}(s)}{4 v_{1}^2(s)}\right)\nonumber\\[10pt]
&+&n^{-\frac{1}{2}}\bigg(\frac{s^2\: v_{1}(s)}{4}-\sqrt{2} v_{1}(s) v_{3}(s)-\frac{s\: v_{1}^2(s)}{2 \sqrt{2}}+\frac{v_{1}^3(s)}{8}-\frac{v_{2}^2(s)}{\sqrt{2}}+s\: v_{3}(s)\nonumber\\[10pt]
&-&\frac{2v_{1}'(s)v_{3}'(s)+(v_{2}'(s))^2}{4v_{1}(s)}+\frac{v_{1}'(s)(v_{1}'(s)v_{3}(s)+2v_{2}(s)v_{2}'(s))}{4v_{1}^2(s)}-\frac{(v_{1}'(s))^2\: v_{2}^2(s)}{4v_{1}^3(s)}\bigg)+O(n^{-\frac{5}{6}}),
\eea
respectively. Here $v_{1}(s),v_{2}(s)$ and $v_{3}(s)$ satisfy the differential equations (\ref{us}),(\ref{vs}) and (\ref{ws}), and the large $s$ asymptotics are given by (\ref{us1}), (\ref{vs1}) and (\ref{ws1}). In addition, $\hat{v}(s):=-\frac{v_{1}(s)}{\sqrt{2}}$ satisfies the Painlev\'{e} XXXIV equation \cite{Bogatskiy,Its}
\be\label{p34}
\hat{v}''(s)=4\hat{v}^2(s)+2s\:\hat{v}(s)+\frac{(\hat{v}'(s))^2}{2\hat{v}(s)}.
\ee
\end{theorem}
\begin{proof}
By changing variable $t_{1}$ to $s$ from the relation $t_{1}=\sqrt{2n}+2^{-\frac{1}{2}}n^{-\frac{1}{6}}s$, and denoting
$$
R_{n}(t_{1})=R_{n}(\sqrt{2n}+2^{-\frac{1}{2}}n^{-\frac{1}{6}}s)=:\tilde{R}_{n}(s),
$$
$$
r_{n}(t_{1})=r_{n}(\sqrt{2n}+2^{-\frac{1}{2}}n^{-\frac{1}{6}}s)=:\tilde{r}_{n}(s),
$$
$$
\sigma_{n}(t_{1})=\sigma_{n}(\sqrt{2n}+2^{-\frac{1}{2}}n^{-\frac{1}{6}}s)=:\tilde{\sigma}_{n}(s),
$$
(\ref{sod}) becomes
\be\label{sod1}
\tilde{R}_{n}''(s)=\frac{\left(\tilde{R}_{n}'(s)\right)^2}{2\tilde{R}_{n}(s)}+\frac{3}{4n^{\frac{1}{3}}}\tilde{R}_{n}^3(s)
-2^{\frac{1}{2}}\left(2n^{\frac{1}{6}}+n^{-\frac{1}{2}}s\right)\tilde{R}_{n}^2(s)+\left(\frac{s^2}{2n^{\frac{2}{3}}}+2s-n^{-\frac{1}{3}}\right)\tilde{R}_{n}(s).
\ee
Let $\check{R}_{n}(s)$ satisfy the quadratic equation by neglecting the derivative terms in (\ref{sod1}),
$$
3\check{R}_{n}^2(s)-4\sqrt{2}\left(n^{-\frac{1}{6}}s+2n^{\frac{1}{2}}\right)\check{R}_{n}(s)
+2n^{-\frac{1}{3}}s^2+8 n^{\frac{1}{3}}s=4.
$$
The solution is
$$
\check{R}_{n}(s)=\frac{1}{3}\left(4\sqrt{2n}+2 \sqrt{2}n^{-\frac{1}{6}} s\pm\sqrt{12+32 n+2n^{-\frac{1}{3}} s^2+8 n^{\frac{1}{3}} s}\right).
$$
As $n\rightarrow\infty,\; t_{1}\rightarrow\infty$, then $w(x,t_{1})\rightarrow A \mathrm{e}^{-x^2}$. It makes $\alpha_{n}(t_{1})\rightarrow 0$ since $w(x,t_{1})$ tends to an even weight function \cite{Chihara}. In view of $R_{n}(t_{1})=2\alpha_{n}(t_{1})$, we see that $\tilde{R}_{n}(s)\rightarrow 0$ as $n\rightarrow\infty$. So we should choose
$$
\check{R}_{n}(s)=\frac{1}{3}\left(4\sqrt{2n}+2 \sqrt{2}n^{-\frac{1}{6}} s-\sqrt{12+32 n+2n^{-\frac{1}{3}} s^2+8 n^{\frac{1}{3}} s}\right).
$$
It follows that as $n\rightarrow\infty$,
$$
\check{R}_{n}(s)=\frac{s}{\sqrt{2}n^{\frac{1}{6}}}-\frac{1}{2\sqrt{2n}}-\frac{s^2}{16\sqrt{2}n^{\frac{5}{6}}}
+O\left(\frac{1}{n^{\frac{7}{6}}}\right).
$$
Hence we suppose that
\be\label{asy}
\tilde{R}_{n}(s)=n^{-\frac{1}{6}}v_{1}(s)+n^{-\frac{1}{2}}v_{2}(s)+n^{-\frac{5}{6}}v_{3}(s)+O(n^{-\frac{7}{6}}).
\ee
Substituting (\ref{asy}) into (\ref{sod1}), we find as $n\rightarrow\infty$,
\be\label{us}
v_{1}''(s)-\frac{(v_{1}'(s))^2}{2 v_{1}(s)}+2 \sqrt{2}\: v_{1}^2(s)-2 s\: v_{1}(s)=0,
\ee
\be\label{vs}
v_{2}''(s)-\frac{v_{1}'(s) v_{2}'(s)}{v_{1}(s)}+\left(\frac{(v_{1}'(s))^2}{2 v_{1}^2(s)}+4 \sqrt{2}\: v_{1}(s)-2s\right)v_{2}(s)+v_{1}(s)=0,
\ee
and
\bea\label{ws}
&&v_{3}''(s)-\frac{v_{1}'(s) v_{3}'(s)}{v_{1}(s)}+\left(\frac{1}{2v_{1}^2(s)}+4 \sqrt{2}\: v_{1}(s)-2s\right)v_{3}(s)-\frac{(v_{2}'(s))^2}{2v_{1}(s)}+\frac{ v_{1}'(s)v_{2}(s) v_{2}'(s)}{v_{1}^2(s)}\nonumber\\[10pt]
&-&\frac{(v_{1}'(s))^2 v_{2}^2(s)}{2v_{1}^3(s)}-\frac{1}{2} s^2\: v_{1}(s)+\sqrt{2} s\: v_{1}^2(s)-\frac{3}{4} v_{1}^3(s)+v_{2}(s)+2 \sqrt{2}\: v_{2}^2(s)=0.
\eea
We obtain the large $s$ asymptotic of $v_{1}(s)$ from (\ref{us}). As $s\rightarrow\infty$,
\be\label{us1}
v_{1}(s)=\frac{s}{\sqrt{2}}+\frac{1}{4 \sqrt{2} s^2}-\frac{9}{8 \sqrt{2} s^5}+\frac{1323}{64 \sqrt{2} s^8}-\frac{108315}{128 \sqrt{2} s^{11}}+o(s^{-11}).
\ee
Substituting (\ref{us1}) into (\ref{vs}), we have the large $s$ behavior of $v_{2}(s)$,
\be\label{vs1}
v_{2}(s)=-\frac{1}{2 \sqrt{2}}+\frac{1}{4 \sqrt{2} s^3}-\frac{45}{16 \sqrt{2} s^6}+\frac{1323}{16 \sqrt{2} s^9}-\frac{1191465}{256 \sqrt{2} s^{12}}+o(s^{-12}).
\ee
Then substituting (\ref{us1}) and (\ref{vs1}) into (\ref{ws}), we obtain the large $s$ asymptotic of $v_{3}(s)$,
\be\label{ws1}
v_{3}(s)=-\frac{s^2}{16 \sqrt{2}}+\frac{73}{256 \sqrt{2} s^4}-\frac{1791}{256 \sqrt{2} s^7}+\frac{686745}{2048 \sqrt{2} s^{10}}-\frac{383291217}{16384 \sqrt{2} s^{13}}+o(s^{-13}).
\ee
In addition, letting $v_{1}(s)=-\sqrt{2}\:\hat{v}(s)$ and substituting into (\ref{us}), we readily see that $\hat{v}(s)$ satisfies the Painlev\'{e} XXXIV equation (\ref{p34}).

Now we consider the large $n$ behavior of $\tilde{r}_{n}(s)$ and $\tilde{\sigma}_{n}(s)$. We find from (\ref{rn}) that
$$
\tilde{r}_{n}(s)=\frac{1}{4}\left[2^{\frac{1}{2}}n^{\frac{1}{6}}\tilde{R}_{n}'(s)-\tilde{R}_{n}^2(s)+2(\sqrt{2n}
+2^{-\frac{1}{2}}n^{-\frac{1}{6}}s)\tilde{R}_{n}(s)\right].
$$
Substituting (\ref{asy}) into the above, we arrive at (\ref{r1}).\\
From (\ref{eq1}) we have
$$
\sigma_{n}(t_{1})=2t_{1}r_{n}(t_{1})-(n+r_{n}(t_{1}))R_{n}(t_{1})-\frac{2r_{n}^2(t_{1})}{R_{n}(t_{1})}.
$$
Replacing $r_{n}(t_{1})$ with the expression of $R_{n}(t_{1})$ from (\ref{rn}), we find
$$
\sigma_{n}(t_{1})=\frac{R_{n}^4(t_{1})-4 t_{1} R_{n}^3(t_{1})+\left(4 t_{1}^2-8 n\right) R_{n}^2(t_{1})-(R_{n}'(t_{1}))^2}{8 R_{n}(t_{1})}.
$$
It follows that
$$
\tilde{\sigma}_{n}(s)=\frac{1}{8}\tilde{R}_{n}^3(s)-\frac{2n^{\frac{1}{2}}+n^{-\frac{1}{6}}s}{2\sqrt{2}}\tilde{R}_{n}^2(s)+s \left( n^{\frac{1}{3}}+4^{-1}n^{-\frac{1}{3}}s\right) \tilde{R}_{n}(s)-\frac{n^{\frac{1}{3}} (\tilde{R}_{n}'(s))^2}{4 \tilde{R}_{n}(s)}.
$$
Using (\ref{asy}), we obtain (\ref{sigma1}). This completes the proof.
\end{proof}
\noindent $\mathbf{Remark.}$ The results of Theorem \ref{cor} coincide with \cite{Its}. See also \cite{Bogatskiy, Xu} on the asymptotics of the recurrence coefficients $\alpha_{n}, \beta_{n}$ and the Painlev\'{e} XXXIV equation. In addition, Perret and Schehr \cite{Perret} also obtained the Painlev\'{e} XXXIV in the study of the gap probability distribution between the first two largest eigenvalues in the Gaussian unitary ensemble.

\begin{proposition}
Assume that $n\rightarrow\infty,\; t_{1}=\sqrt{2n}+2^{-\frac{1}{2}}n^{-\frac{1}{6}}s$ and $s$ is fixed. Then as $n\rightarrow\infty$,  $\tilde{P}_{n}(z):=P_{n}(z,\sqrt{2n}+2^{-\frac{1}{2}}n^{-\frac{1}{6}}s)$ satisfies the Hermite's differential equation
$$
\tilde{P}_{n}''(z)-2z\: \tilde{P}_{n}'(z)+2n\: \tilde{P}_{n}(z)=0.
$$
\end{proposition}
\begin{proof}
By changing variable $t_{1}$ to $s$, (\ref{ode1}) becomes
\bea
&&\tilde{P}_n''(z)+\tilde{P}_n'(z) \left(\frac{\tilde{R}_n(s)}{(z-\tilde{s}) (2z-2 \tilde{s}+\tilde{R}_n(s)}-2z\right)+\tilde{P}_n(z)\bigg(2 n-\frac{2^{\frac{1}{2}}n^{\frac{1}{6}}\tilde{R}_n'(s)-\tilde{R}_n^2(s)+2 \tilde{s} \tilde{R}_n(s)}{4 (z-\tilde{s})^2}\nonumber\\[10pt]
&+&\frac{\tilde{R}_n(s) \left(2^{\frac{1}{2}}n^{\frac{1}{6}}\tilde{R}_n'(s)-\tilde{R}_n^2(s)+2 \tilde{s} \tilde{R}_n(s)\right)}{4 (z-\tilde{s})^2 (2z-2 \tilde{s}+\tilde{R}_n(s))}+\frac{2n^{\frac{1}{3}}(\tilde{R}_n'(s))^2-\tilde{R}_n^4(s)+4 \tilde{s} \tilde{R}_n^3(s)+(8 n-4 \tilde{s}^2) \tilde{R}_n^2(s)}{8(z-\tilde{s}) \tilde{R}_n(s)}\bigg)=0,\nonumber
\eea
where $\tilde{R}_{n}(s)=R_{n}(\sqrt{2n}+2^{-\frac{1}{2}}n^{-\frac{1}{6}}s)$ and $\tilde{s}:=\sqrt{2n}+2^{-\frac{1}{2}}n^{-\frac{1}{6}}s$.
\\
Substituting (\ref{asy}) into the above, we find as $n\rightarrow\infty$, the coefficient of $\tilde{P}_n'(z)$ is
$$
-2z+\frac{v_{1}(s)}{4n^{\frac{7}{6}}}+O(n^{-\frac{3}{2}})
$$
and the coefficient of $\tilde{P}_n(z)$ is
$$
2n-\frac{(v_{1}'(s))^2-4 s\: v_{1}^2(s)+2 \sqrt{2} v_{1}^3(s)}{4 \sqrt{2}n^{\frac{1}{3}} v_{1}(s)}+O(n^{-\frac{2}{3}}).
$$
It follows that as $n\rightarrow\infty$,
$$
\tilde{P}_{n}''(z)-2z\: \tilde{P}_{n}'(z)+2n\: \tilde{P}_{n}(z)=0.
$$
\end{proof}

\section{Gaussian Weight with Two Jumps}
In this section, we suppose $B_{1}\neq0$ and $B_{2}\neq0$, which corresponds to the Gaussian weight with two jumps
$$
w(x,t_{1},t_{2}):=\mathrm{e}^{-x^2}(A+B_{1}\theta(x-t_{1})+B_{2}\theta(x-t_{2})),\;\;t_{1}<t_{2}.
$$
This is an extension of \cite{Chen2005} which considered the Gaussian weight with a single jump. For example,
$$
w(x,t_{1},t_{2})=\mathrm{e}^{-x^2}\cdot\left\{
\begin{aligned}
&\mathrm{e}^{\mu},&x<t_{1};\\
&1,&t_{1}<x<t_{2};\\
&\mathrm{e}^{-\mu},&x>t_{2},
\end{aligned}
\right.
$$
corresponds to $A=\mathrm{e}^{\mu},\;B_{1}=1-\mathrm{e}^{\mu},\;B_{2}=\mathrm{e}^{-\mu}-1$.

\subsection{Ladder Operators}
Let $P_{n}(x,t_{1},t_{2})$ be the monic polynomials of degree $n$ orthogonal with respect to the weight function $w(x,t_{1},t_{2})$,
\be\label{ort}
\int_{-\infty}^{\infty}P_{j}(x,t_{1},t_{2})P_{k}(x,t_{1},t_{2})w(x,t_{1},t_{2})dx=h_{j}(t_{1},t_{2})\delta_{jk},\;\;j,k=0,1,2,\ldots,
\ee
where
$$
w(x,t_{1},t_{2}):=w_{0}(x)(A_{1}+B_{1}\theta(x-t_{1})+B_{2}\theta(x-t_{2})),\;\;w_{0}(x):=\mathrm{e}^{-\mathrm{v}_{0}(x)},\;\;\mathrm{v}_{0}(x)=x^2.
$$
Similarly as the previous section, the monic polynomials $P_{n}(x,t_{1},t_{2})$ can be written in the form
\be\label{expant}
P_{n}(x,t_{1},t_{2})=x^{n}+\mathrm{p}(n,t_{1},t_{2})x^{n-1}+\cdots+P_{n}(0,t_{1},t_{2}),
\ee
and the recurrence relation reads
$$
xP_{n}(x,t_{1},t_{2})=P_{n+1}(x,t_{1},t_{2})+\alpha_{n}(t_{1},t_{2})P_{n}(x,t_{1},t_{2})+\beta_{n}(t_{1},t_{2})P_{n-1}(x,t_{1},t_{2})
$$
with the initial conditions
$$
P_{0}(x,t_{1},t_{2})=1,\;\;\beta_{0}(t_{1},t_{2})P_{-1}(x,t_{1},t_{2})=0.
$$
We also have the expressions of $\alpha_{n}(t_{1},t_{2})$ and $\beta_{n}(t_{1},t_{2})$:
\be\label{alpt}
\alpha_{n}(t_{1},t_{2})=\mathrm{p}(n,t_{1},t_{2})-\mathrm{p}(n+1,t_{1},t_{2}),
\ee
$$
\beta_{n}(t_{1},t_{2})=\frac{h_{n}(t_{1},t_{2})}{h_{n-1}(t_{1},t_{2})}.
$$
A telescopic sum of (\ref{alpt}) gives
$$
\sum_{j=0}^{n-1}\alpha_{j}(t_{1},t_{2})=-\mathrm{p}(n,t_{1},t_{2}).
$$
We also denote
$$
D_{n}(t_{1},t_{2}):=D_{n}[w]=\frac{1}{n!}\int_{(-\infty,\infty)^{n}}\prod_{1\leq j<k\leq n}\left(x_{k}-x_{j}\right)^{2}
\prod_{j=1}^{n}w(x_{j},t_{1},t_{2})dx_{j}
$$
and we have
\be\label{dn}
D_{n}(t_{1},t_{2})=\det\left(\int_{-\infty}^{\infty}x^{i+j}w(x,t_{1},t_{2})dx\right)_{i,j=0}^{n-1}=\prod_{j=0}^{n-1}h_{j}(t_{1},t_{2}).
\ee
To simplify notations, we suppress the $t_{1}, t_{2}$ dependence in $P_{n}(x)$, $h_{n}$, $\alpha_{n}$ and $\beta_{n}$ in the following discussions.
From Lemma 1 and Remark 2 in \cite{Basor2009}, we have the following theorem.
\begin{theorem}
The lowering and raising operators for monic polynomials orthogonal with respect to $w(x,t_{1},t_{2})$ are
$$
P_{n}'(z)=\beta_{n}A_{n}(z)P_{n-1}(z)-B_{n}(z)P_{n}(z),
$$
$$
P_{n-1}'(z)=(B_{n}(z)+\mathrm{v}_{0}'(z))P_{n-1}(z)-A_{n-1}(z)P_{n}(z),
$$
where
\be\label{anz2}
A_{n}(z)=2+\frac{R_{n,1}(t_{1},t_{2})}{z-t_{1}}+\frac{R_{n,2}(t_{1},t_{2})}{z-t_{2}},
\ee
\be\label{bnz2}
B_{n}(z)=\frac{r_{n,1}(t_{1},t_{2})}{z-t_{1}}+\frac{r_{n,2}(t_{1},t_{2})}{z-t_{2}},
\ee
and
$$
R_{n,1}(t_{1},t_{2})=\frac{B_{1}P_{n}^2(t_{1})\mathrm{e}^{-t_{1}^2}}{h_{n}},\;\;\;\;R_{n,2}(t_{1},t_{2})=\frac{B_{2}P_{n}^2(t_{2})\mathrm{e}^{-t_{1}^2}}{h_{n}},
$$
$$
r_{n,1}(t_{1},t_{2})=\frac{B_{1}P_{n}(t_{1})P_{n-1}(t_{1})\mathrm{e}^{-t_{1}^2}}{h_{n-1}},\;\;\;\;
r_{n,2}(t_{1},t_{2})=\frac{B_{2}P_{n}(t_{2})P_{n-1}(t_{2})\mathrm{e}^{-t_{2}^2}}{h_{n-1}}.
$$
Here $P_{n}(t_{1})=P_{n}(x,t_{1},t_{2})|_{x=t_{1}},\;P_{n}(t_{2})=P_{n}(x,t_{1},t_{2})|_{x=t_{2}}$.
\end{theorem}
From \cite{Basor2009, Chen2005}, we see that Lemma \ref{s1s2}, \ref{s2p} and \ref{de} are still valid for the weight with two jumps.
Substituting (\ref{anz2}) and (\ref{bnz2}) into ($S_{1}$), we obtain
$$
\frac{r_{n+1,1}(t_{1},t_{2})+r_{n,1}(t_{1},t_{2})}{z-t_{1}}+\frac{r_{n+1,2}(t_{1},t_{2})+r_{n,2}(t_{1},t_{2})}{z-t_{2}}
=\frac{(z-\alpha_{n})R_{n,1}(t_{1},t_{2})}{z-t_{1}}+\frac{(z-\alpha_{n})R_{n,2}(t_{1},t_{2})}{z-t_{2}}-2\alpha_{n}.
$$
It follows that
\be\label{s1e}
R_{n,1}(t_{1},t_{2})+R_{n,2}(t_{1},t_{2})=2\alpha_{n},
\ee
$$
r_{n+1,1}(t_{1},t_{2})+r_{n,1}(t_{1},t_{2})=(t_{1}-\alpha_{n})R_{n,1}(t_{1},t_{2}),
$$
and
$$
r_{n+1,2}(t_{1},t_{2})+r_{n,2}(t_{1},t_{2})=(t_{2}-\alpha_{n})R_{n,2}(t_{1},t_{2}).
$$
Similarly, plugging (\ref{anz2}) and (\ref{bnz2}) into ($S_{2}'$), we obtain
\bea
&&\left(\frac{r_{n,1}(t_{1},t_{2})}{z-t_{1}}+\frac{r_{n,2}(t_{1},t_{2})}{z-t_{2}}\right)^2
+2z\left(\frac{r_{n,1}(t_{1},t_{2})}{z-t_{1}}+\frac{r_{n,2}(t_{1},t_{2})}{z-t_{2}}\right)+\frac{\sum_{j=0}^{n-1}R_{j,1}(t_{1},t_{2})}{z-t_{1}}
+\frac{\sum_{j=0}^{n-1}R_{j,2}(t_{1},t_{2})}{z-t_{2}}\nonumber\\
&+&2n=\beta_{n}\left(2+\frac{R_{n,1}(t_{1},t_{2})}{z-t_{1}}+\frac{R_{n,2}(t_{1},t_{2})}{z-t_{2}}\right)
\left(2+\frac{R_{n-1,1}(t_{1},t_{2})}{z-t_{1}}+\frac{R_{n-1,2}(t_{1},t_{2})}{z-t_{2}}\right).
\eea
It implies the following equalities:
\be\label{s2p1}
\beta_{n}=\frac{n+r_{n,1}(t_{1},t_{2})+r_{n,2}(t_{1},t_{2})}{2},
\ee
\be\label{s2p2}
r_{n,1}^{2}(t_{1},t_{2})=\beta_{n}R_{n,1}(t_{1},t_{2})R_{n-1,1}(t_{1},t_{2}),
\ee
\be\label{s2p3}
r_{n,2}^{2}(t_{1},t_{2})=\beta_{n}R_{n,2}(t_{1},t_{2})R_{n-1,2}(t_{1},t_{2}),
\ee
\bea\label{s2p4}
&&\frac{2r_{n,1}(t_{1},t_{2})r_{n,2}(t_{1},t_{2})}{t_{1}-t_{2}}+2t_{1}r_{n,1}(t_{1},t_{2})+\sum_{j=0}^{n-1}R_{j,1}(t_{1},t_{2})\nonumber\\
&=&\beta_{n}\left(\frac{R_{n,1}(t_{1},t_{2})R_{n-1,2}(t_{1},t_{2})+R_{n,2}(t_{1},t_{2})R_{n-1,1}(t_{1},t_{2})}{t_{1}-t_{2}}
+2R_{n,1}(t_{1},t_{2})+2R_{n-1,1}(t_{1},t_{2})\right),\nonumber\\
\eea
\bea\label{s2p5}
&&\frac{2r_{n,1}(t_{1},t_{2})r_{n,2}(t_{1},t_{2})}{t_{2}-t_{1}}+2t_{2}r_{n,2}(t_{1},t_{2})+\sum_{j=0}^{n-1}R_{j,2}(t_{1},t_{2})\nonumber\\
&=&\beta_{n}\left(\frac{R_{n,1}(t_{1},t_{2})R_{n-1,2}(t_{1},t_{2})+R_{n,2}(t_{1},t_{2})R_{n-1,1}(t_{1},t_{2})}{t_{2}-t_{1}}
+2R_{n,2}(t_{1},t_{2})+2R_{n-1,2}(t_{1},t_{2})\right).\nonumber\\
\eea
The sum of (\ref{s2p4}) and (\ref{s2p5}) gives
\bea\label{s2p6}
&&2t_{1}r_{n,1}(t_{1},t_{2})+2t_{2}r_{n,2}(t_{1},t_{2})+\sum_{j=0}^{n-1}R_{j,1}(t_{1},t_{2})+\sum_{j=0}^{n-1}R_{j,2}(t_{1},t_{2})\nonumber\\
&=&2\beta_{n}\left(R_{n,1}(t_{1},t_{2})+R_{n,2}(t_{1},t_{2})+R_{n-1,1}(t_{1},t_{2})+R_{n-1,2}(t_{1},t_{2})\right).
\eea

\subsection{Toda Evolution in $t_{1}$ and $t_{2}$}
Taking a derivative with respect to $t_{1}$ and $t_{2}$ in the equation
$$
\int_{-\infty}^{\infty}P_{n}^2(x,t_{1},t_{2})\mathrm{e}^{-x^2}(A_{1}+B_{1}\theta(x-t_{1})+B_{2}\theta(x-t_{2}))dx=h_{n}(t_{1},t_{2}),\;\;n=0,1,2,\ldots,
$$
respectively, we obtain
$$
\partial_{t_{1}}h_{n}(t_{1},t_{2})=-B_{1}\:P_{n}^{2}(t_{1})\mathrm{e}^{-t_{1}^2}
$$
and
$$
\partial_{t_{2}}h_{n}(t_{1},t_{2})=-B_{2}\:P_{n}^{2}(t_{2})\mathrm{e}^{-t_{2}^2}.
$$
It follows that
\be\label{h1}
\partial_{t_{1}}\ln h_{n}(t_{1},t_{2})=-R_{n,1}(t_{1},t_{2}),
\ee
\be\label{h2}
\partial_{t_{2}}\ln h_{n}(t_{1},t_{2})=-R_{n,2}(t_{1},t_{2}),
\ee
and
$$
\partial_{t_{1}}[\ln\beta_{n}(t_{1},t_{2})]=\partial_{t_{1}}\ln h_{n}(t_{1},t_{2})-\partial_{t_{1}}\ln h_{n-1}(t_{1},t_{2})=R_{n-1,1}(t_{1},t_{2})-R_{n,1}(t_{1},t_{2}),
$$
$$
\partial_{t_{2}}[\ln\beta_{n}(t_{1},t_{2})]=\partial_{t_{2}}\ln h_{n}(t_{1},t_{2})-\partial_{t_{2}}\ln h_{n-1}(t_{1},t_{2})=R_{n-1,2}(t_{1},t_{2})-R_{n,2}(t_{1},t_{2}).
$$
Hence,
\be\label{eqr1}
\partial_{t_{1}}\beta_{n}(t_{1},t_{2})=\beta_{n}R_{n-1,1}(t_{1},t_{2})-\beta_{n}R_{n,1}(t_{1},t_{2})
=\frac{r_{n,1}^2(t_{1},t_{2})}{R_{n,1}(t_{1},t_{2})}-\beta_{n}R_{n,1}(t_{1},t_{2}),
\ee
\be\label{eqr2}
\partial_{t_{2}}\beta_{n}(t_{1},t_{2})=\beta_{n}R_{n-1,2}(t_{1},t_{2})-\beta_{n}R_{n,2}(t_{1},t_{2})
=\frac{r_{n,2}^2(t_{1},t_{2})}{R_{n,2}(t_{1},t_{2})}-\beta_{n}R_{n,2}(t_{1},t_{2}).
\ee
On the other hand, differentiating with respect to $t_{1}$ and $t_{2}$ in the equation
$$
\int_{-\infty}^{\infty}P_{n}(x,t_{1},t_{2})P_{n-1}(x,t_{1},t_{2})\mathrm{e}^{-x^2}(A_{1}+B_{1}\theta(x-t_{1})+B_{2}\theta(x-t_{2}))dx=0,\;\;n=0,1,2,\ldots,
$$
respectively gives
\be\label{dp1}
\partial_{t_{1}}\mathrm{p}(n,t_{1},t_{2})=r_{n,1}(t_{1},t_{2}),
\ee
and
\be\label{dp2}
\partial_{t_{2}}\mathrm{p}(n,t_{1},t_{2})=r_{n,2}(t_{1},t_{2}).
\ee
Now we have the two variables' Toda equations on $\alpha_{n}$ and $\beta_{n}$.
\begin{proposition}
\be\label{toda11}
\partial_{t_{1}}\beta_{n}+\partial_{t_{2}}\beta_{n}=2\beta_{n}(\alpha_{n-1}-\alpha_{n}),
\ee
\be\label{toda21}
\partial_{t_{1}}\alpha_{n}+\partial_{t_{2}}\alpha_{n}=2(\beta_{n}-\beta_{n+1})+1.
\ee
\end{proposition}
\begin{proof}
The sum of (\ref{eqr1}) and (\ref{eqr2}) gives
$$
\partial_{t_{1}}\beta_{n}(t_{1},t_{2})+\partial_{t_{2}}\beta_{n}(t_{1},t_{2})=\beta_{n}(R_{n-1,1}(t_{1},t_{2})+R_{n-1,2}(t_{1},t_{2})
-R_{n,1}(t_{1},t_{2})-R_{n,2}(t_{1},t_{2})).
$$
Using (\ref{s1e}), we arrive at (\ref{toda11}).
From (\ref{alpt}), (\ref{dp1}) and (\ref{dp2}) we have
$$
\partial_{t_{1}}\alpha_{n}+\partial_{t_{2}}\alpha_{n}=r_{n,1}(t_{1},t_{2})+r_{n,2}(t_{1},t_{2})-r_{n+1,1}(t_{1},t_{2})-r_{n+1,2}(t_{1},t_{2}).
$$
With the aid of (\ref{s2p1}), we readily obtain (\ref{toda21}).
\end{proof}

\subsection{Generalized Jimbo-Miwa-Okamoto $\sigma$ Form of Painlev\'{e} IV}
We define a quantity related to the Hankel determinant $D_{n}(t_{1},t_{2})$,
$$
\sigma_{n}(t_{1},t_{2}):=\partial_{t_{1}}\ln D_{n}(t_{1},t_{2})+\partial_{t_{2}}\ln D_{n}(t_{1},t_{2}).
$$
It is easy from (\ref{dn}), (\ref{h1}) and (\ref{h2}) to see that
$$
\sigma_{n}(t_{1},t_{2})=-\sum_{j=0}^{n-1}R_{j,1}(t_{1},t_{2})-\sum_{j=0}^{n-1}R_{j,2}(t_{1},t_{2}).
$$
Then (\ref{s2p6}) becomes
\be\label{s2p7}
2t_{1}r_{n,1}(t_{1},t_{2})+2t_{2}r_{n,2}(t_{1},t_{2})-\sigma_{n}(t_{1},t_{2})
=4\beta_{n}\left(R_{n,1}(t_{1},t_{2})+R_{n,2}(t_{1},t_{2})\right)+2\partial_{t_{1}}\beta_{n}+2\partial_{t_{2}}\beta_{n},
\ee
where we have made use of (\ref{eqr1}) and (\ref{eqr2}).

From (\ref{s1e}) we have
$$
\sigma_{n}(t_{1},t_{2})=-2\sum_{j=0}^{n-1}\alpha_{j}=2\mathrm{p}(n,t_{1},t_{2}).
$$
It follows from (\ref{dp1}) and (\ref{dp2}) that
\be\label{rn1}
\partial_{t_{1}}\sigma_{n}(t_{1},t_{2})=2r_{n,1}(t_{1},t_{2}),
\ee
\be\label{rn2}
\partial_{t_{2}}\sigma_{n}(t_{1},t_{2})=2r_{n,2}(t_{1},t_{2}).
\ee
Then we see from (\ref{s2p1}) that
\be\label{ben}
\beta_{n}=\frac{2n+\partial_{t_{1}}\sigma_{n}(t_{1},t_{2})+\partial_{t_{2}}\sigma_{n}(t_{1},t_{2})}{4}.
\ee
Substituting (\ref{rn1}), (\ref{ben}) into (\ref{eqr1}) and (\ref{rn2}), (\ref{ben}) into (\ref{eqr2}), we obtain
\be\label{sig1}
\partial_{t_{1}}^2\sigma_{n}(t_{1},t_{2})+\partial_{t_{1}}\partial_{t_{2}}\sigma_{n}(t_{1},t_{2})
=\frac{(\partial_{t_{1}}\sigma_{n}(t_{1},t_{2}))^2}{R_{n,1}(t_{1},t_{2})}-(2n+\partial_{t_{1}}\sigma_{n}(t_{1},t_{2})
+\partial_{t_{2}}\sigma_{n}(t_{1},t_{2}))R_{n,1}(t_{1},t_{2})
\ee
and
\be\label{sig2}
\partial_{t_{2}}^2\sigma_{n}(t_{1},t_{2})+\partial_{t_{2}}\partial_{t_{1}}\sigma_{n}(t_{1},t_{2})
=\frac{(\partial_{t_{2}}\sigma_{n}(t_{1},t_{2}))^2}{R_{n,2}(t_{1},t_{2})}-(2n+\partial_{t_{1}}\sigma_{n}(t_{1},t_{2})
+\partial_{t_{2}}\sigma_{n}(t_{1},t_{2}))R_{n,2}(t_{1},t_{2})
\ee
respectively.
\\
We regard (\ref{sig1}) and (\ref{sig2}) as quadratic equations on $R_{n,1}(t_{1},t_{2})$ and $R_{n,2}(t_{1},t_{2})$, respectively. The solutions are
\be\label{sol1}
R_{n,1}(t_{1},t_{2})=\frac{-\partial_{t_{1}}^2\sigma_{n}-\partial_{t_{1}}\partial_{t_{2}}\sigma_{n}\pm
\sqrt{\Delta_{1}}}{2(2n+\partial_{t_{1}}\sigma_{n}+\partial_{t_{2}}\sigma_{n})}
\ee
and
\be\label{sol2}
R_{n,2}(t_{1},t_{2})=\frac{-\partial_{t_{2}}^2\sigma_{n}-\partial_{t_{2}}\partial_{t_{1}}\sigma_{n}\pm
\sqrt{\Delta_{2}}}{2(2n+\partial_{t_{1}}\sigma_{n}+\partial_{t_{2}}\sigma_{n})},
\ee
where we write $\sigma_{n}(t_{1},t_{2})$ as $\sigma_{n}$ for short, and
$$
\Delta_{1}:=(\partial_{t_{1}}^2\sigma_{n}+\partial_{t_{1}}\partial_{t_{2}}\sigma_{n})^2
+4(\partial_{t_{1}}\sigma_{n})^2(2n+\partial_{t_{1}}\sigma_{n}+\partial_{t_{2}}\sigma_{n}),
$$
$$
\Delta_{2}:=(\partial_{t_{2}}^2\sigma_{n}+\partial_{t_{2}}\partial_{t_{1}}\sigma_{n})^2
+4(\partial_{t_{2}}\sigma_{n})^2(2n+\partial_{t_{1}}\sigma_{n}+\partial_{t_{2}}\sigma_{n}).
$$
Substituting (\ref{rn1}), (\ref{rn2}), (\ref{ben}), (\ref{sol1}) and (\ref{sol2}) into (\ref{s2p7}), we obtain a second order differential equation on $\sigma_{n}$,
$$
(2t_{1}\partial_{t_{1}}\sigma_{n}+2t_{2}\partial_{t_{2}}\sigma_{n}-2\sigma_{n})^2=\Delta_{1}+\Delta_{2}\pm 2\sqrt{\Delta_{1}\Delta_{2}},
$$
which is equivalent to
$$
[(2t_{1}\partial_{t_{1}}\sigma_{n}+2t_{2}\partial_{t_{2}}\sigma_{n}-2\sigma_{n})^2-\Delta_{1}-\Delta_{2}]^2=4\Delta_{1}\Delta_{2}.
$$
We summarize it into the following theorem.
\begin{theorem}
$\sigma_{n}:=\sigma_{n}(t_{1},t_{2})$ satisfies the second order partial differential equation:
\bea\label{equ}
&&[(2t_{1}\partial_{t_{1}}\sigma_{n}+2t_{2}\partial_{t_{2}}\sigma_{n}-2\sigma_{n})^2
-(\partial_{t_{1}}^2\sigma_{n}+\partial_{t_{1}}\partial_{t_{2}}\sigma_{n})^2-(\partial_{t_{2}}^2\sigma_{n}+\partial_{t_{2}}\partial_{t_{1}}\sigma_{n})^2
\nonumber\\
&-&4(\partial_{t_{1}}\sigma_{n})^2(2n+\partial_{t_{1}}\sigma_{n}+\partial_{t_{2}}\sigma_{n})
-4(\partial_{t_{2}}\sigma_{n})^2(2n+\partial_{t_{1}}\sigma_{n}+\partial_{t_{2}}\sigma_{n})]^2\nonumber\\
&=&4[(\partial_{t_{1}}^2\sigma_{n}+\partial_{t_{1}}\partial_{t_{2}}\sigma_{n})^2
+4(\partial_{t_{1}}\sigma_{n})^2(2n+\partial_{t_{1}}\sigma_{n}+\partial_{t_{2}}\sigma_{n})]\nonumber\\
&&[(\partial_{t_{2}}^2\sigma_{n}+\partial_{t_{2}}\partial_{t_{1}}\sigma_{n})^2
+4(\partial_{t_{2}}\sigma_{n})^2(2n+\partial_{t_{1}}\sigma_{n}+\partial_{t_{2}}\sigma_{n})].
\eea
\end{theorem}
Actually, (\ref{equ}) is a two variables' generalization of the Jimbo-Miwa-Okamoto $\sigma$ form of the Painlev\'{e} IV.
If $\sigma_{n}$ is independent of $t_{2}$, then (\ref{equ}) is reduced to
\be\label{p41}
(\partial_{t_{1}}^2\sigma_{n})^2=4(t_{1}\partial_{t_{1}}\sigma_{n}-\sigma_{n})^2-4(\partial_{t_{1}}\sigma_{n})^2(2n+\partial_{t_{1}}\sigma_{n}),
\ee
which is the Jimbo-Miwa-Okamoto $\sigma$ form of the Painlev\'{e} IV. Note that (\ref{p41}) is the same with (\ref{jmo}). Similarly, if $\sigma_{n}$ is independent of $t_{1}$, then (\ref{equ}) becomes
$$
(\partial_{t_{2}}^2\sigma_{n})^2=4(t_{2}\partial_{t_{2}}\sigma_{n}-\sigma_{n})^2-4(\partial_{t_{2}}\sigma_{n})^2(2n+\partial_{t_{2}}\sigma_{n}),
$$
which is also the Jimbo-Miwa-Okamoto $\sigma$ form of the Painlev\'{e} IV. The results of this section coincide with \cite{Basor2012}, which is a special case of our problem with $A=0,\;B_{1}=1,\;B_{2}=-1$.

\begin{corollary}
Assume that $n\rightarrow\infty$, $t_1=\sqrt{2n}+2^{-\frac{1}{2}}n^{-\frac{1}{6}}s_1,\; t_2=\sqrt{2n}+2^{-\frac{1}{2}}n^{-\frac{1}{6}}s_2$ and $s_1, s_2$ are fixed. Then as $n\rightarrow\infty$, $\tilde{\sigma}_{n}(s_1,s_2):=\sigma_{n}(\sqrt{2n}+2^{-\frac{1}{2}}n^{-\frac{1}{6}}s_1,\sqrt{2n}+2^{-\frac{1}{2}}n^{-\frac{1}{6}}s_2)$ satisfies the following second order partial differential equation:
$$
\partial_{s_1}\tilde{\sigma}_{n}\left(\partial_{s_2}^{2}\tilde{\sigma}_{n}+\partial_{s_2}\partial_{s_1}\tilde{\sigma}_{n}\right)^{2}+
\partial_{s_2}\tilde{\sigma}_{n}\left(\partial_{s_1}^{2}\tilde{\sigma}_{n}+\partial_{s_1}\partial_{s_2}\tilde{\sigma}_{n}\right)^{2}=
4\partial_{s_1}\tilde{\sigma}_{n}\:\partial_{s_2}\tilde{\sigma}_{n}\left(s_1\partial_{s_1}\tilde{\sigma}_{n}
+s_2\partial_{s_2}\tilde{\sigma}_{n}-\tilde{\sigma}_{n}\right).
$$
\end{corollary}
\begin{proof}
Substituting $t_1=\sqrt{2n}+2^{-\frac{1}{2}}n^{-\frac{1}{6}}s_1$ and $t_2=\sqrt{2n}+2^{-\frac{1}{2}}n^{-\frac{1}{6}}s_2$ into (\ref{equ}), and noting that
$\tilde{\sigma}_{n}(s_1,s_2):=\sigma_{n}(\sqrt{2n}+2^{-\frac{1}{2}}n^{-\frac{1}{6}}s_1,\sqrt{2n}+2^{-\frac{1}{2}}n^{-\frac{1}{6}}s_2)$, we obtain
\bea
&&\left(\partial_{s_1}\tilde{\sigma}_{n}+\partial_{s_2}\tilde{\sigma}_{n}\right)\Big[\partial_{s_1}\tilde{\sigma}_{n}\left(\partial_{s_2}^{2}\tilde{\sigma}_{n}
+\partial_{s_2}\partial_{s_1}\tilde{\sigma}_{n}\right)^{2}+
\partial_{s_2}\tilde{\sigma}_{n}\left(\partial_{s_1}^{2}\tilde{\sigma}_{n}+\partial_{s_1}\partial_{s_2}\tilde{\sigma}_{n}\right)^{2}\nonumber\\
&-&4\partial_{s_1}\tilde{\sigma}_{n}\:\partial_{s_2}\tilde{\sigma}_{n}\left(s_1\partial_{s_1}\tilde{\sigma}_{n}
+s_2\partial_{s_2}\tilde{\sigma}_{n}-\tilde{\sigma}_{n}\right)\Big]+2\sqrt{2}n^{-\frac{1}{6}}\partial_{s_1}\tilde{\sigma}_{n}\:\partial_{s_2}\tilde{\sigma}_{n}
\left(\partial_{s_1}\tilde{\sigma}_{n}+\partial_{s_2}\tilde{\sigma}_{n}\right)^{3}+O(n^{-\frac{2}{3}})=0.\nonumber
\eea
Since the terms of $O(n^{-\frac{2}{3}})$ are very complicated, we do not write down the detailed result here. Letting $n\rightarrow\infty$ and disregarding the terms of $O(n^{-\frac{1}{6}})$ and higher order terms, we have
$$
\partial_{s_1}\tilde{\sigma}_{n}\left(\partial_{s_2}^{2}\tilde{\sigma}_{n}+\partial_{s_2}\partial_{s_1}\tilde{\sigma}_{n}\right)^{2}+
\partial_{s_2}\tilde{\sigma}_{n}\left(\partial_{s_1}^{2}\tilde{\sigma}_{n}+\partial_{s_1}\partial_{s_2}\tilde{\sigma}_{n}\right)^{2}=
4\partial_{s_1}\tilde{\sigma}_{n}\:\partial_{s_2}\tilde{\sigma}_{n}\left(s_1\partial_{s_1}\tilde{\sigma}_{n}
+s_2\partial_{s_2}\tilde{\sigma}_{n}-\tilde{\sigma}_{n}\right).
$$
\end{proof}

In the end, we mention a special case, namely $B_{1}=-B_{2}$ and $t_{1}=-t_{2}$. This case makes the weight $w(x,t_1,t_2)$ an even function, and is a generalization of \cite{Lyu} on the symmetric gap probability distribution of the Gaussian unitary ensemble.

\section{Conclusion}
We study the Hankel determinants for a Gaussian weight with one and two jumps. For the single jump case, by using the ladder operator approach, we obtain three auxiliary quantities $r_{n}(t_{1})$, $R_{n}(t_{1})$, $\sigma_{n}(t_{1})$, related to the Hankel determinant $D_{n}(t_{1})$. From (\ref{sigma}), (\ref{rnt}) and (\ref{eq4}) we know the relations are
$$
\sigma_{n}(t_{1})=\frac{d}{dt_{1}}\ln D_{n}(t_{1}),
$$
$$
r_{n}(t_{1})=\frac{1}{2}\sigma_{n}'(t_{1}),
$$
and
$$
R_{n}(t_{1})=\frac{2t_{1}\sigma_{n}'(t_{1})-2\sigma_{n}(t_{1})-\sigma_{n}''(t_{1})}{4n+2\sigma_{n}'(t_{1})}.
$$
We show that $\sigma_{n}(t_{1})$ satisfies both the continuous and discrete $\sigma$ form of the Painlev\'{e} IV. $y(t_{1})=R_{n}(-t_{1})$ satisfies a Painlev\'{e} IV and $v(t_{1})=-2r_{n}(t_{1})-\frac{2n}{3}$ satisfies a Chazy II. $R_{n}(t_{1})$ and $r_{n}(t_{1})$ also satisfy a second order non-linear difference equations respectively. Moreover, we consider the large $n$ behavior of the corresponding monic orthogonal polynomials and prove that they satisfy the biconfluent Heun equation. We also consider the large $n$ asymptotics of the auxiliary quantities under a double scaling, which gives rise to the Painlev\'{e} XXXIV equation. For the general case with two jumps, we also use the ladder operator approach to obtain a partial differential equation to describe the Hankel determinant $D_{n}(t_{1},t_{2})$, which is a two variables' generalization of the Jimbo-Miwa-Okamoto $\sigma$ form of the Painlev\'{e} IV.

\section*{Acknowledgments}
This work of Chao Min was supported by the Scientific Research Funds of Huaqiao University under grant number 600005-Z17Y0054.
Yang Chen was supported by the Macau Science and Technology Development Fund under grant numbers FDCT 130/2014/A3, FDCT 023/2017/A1 and the University of Macau through MYRG 2014-00011-FST, MYRG 2014-00004-FST.

\end{document}